\documentclass{amsart}
\usepackage{times, stmaryrd, verbatim, enumerate, simplewick}
\usepackage{ amssymb, tikz-cd,mathrsfs} 
\DeclareMathAlphabet{\mathpzc}{OT1}{pzc}{m}{it}

\RequirePackage{color}[1999/02/16 v1.0i Standard LaTeX Color (DPC)]

\newif\ifcomments

\let\newComments\newKibitzer

\newComments\GM{GM}{red}
\newComments\SB{SB}{blue}

\newtheorem{Thm}{Theorem}

\newtheorem{Lem}{Lemma}

\newtheorem{corollary}{Corollary}

\newtheorem{rmk}{Remark}[section]

\numberwithin{equation}{section}

\newcommand{\mr}{\mathrm}

\newcommand{\tr}{\mr{tr}\,}

\usepackage{hyperref}
\begin{document}

\title{The classical dynamic symmetry for the $\mathrm{Sp}(1)$-Kepler problems}
\author{Sofiane Bouarroudj}

\address{Division of Science and Mathematics, New York University Abu Dhabi, Po Box 129188, Abu Dhabi, United Arab Emirates.}
\email{sofiane.bouarroudj@nyu.edu}

\author{Guowu Meng}
\address{Department of Mathematics, Hong Kong Univ. of Sci. and
Tech., Clear Water Bay, Kowloon, Hong Kong.}

\email{mameng@ust.hk}
\thanks{The authors were supported by the Hong Hong Research Grants Council under RGC Project No. 16304014; SB was also supported by the grant NYUAD-065.}


\date{\today}



\begin{abstract}
A Poisson realization of the simple real Lie algebra $\mathfrak {so}^*(4n)$ on the phase space of each $\mathrm {Sp}(1)$-Kepler problem is exhibited. As a consequence one obtains the Laplace-Runge-Lenz vector for each classical $\mathrm{Sp}(1)$-Kepler problem. The verification of these Poisson realizations is greatly simplified via an idea due to A. Weinstein.  The totality of these Poisson realizations is shown to be equivalent to the canonical Poisson realization of $\mathfrak {so}^*(4n)$ on the Poisson manifold $T^*\mathbb H_*^n/\mathrm{Sp}(1)$. (Here $\mathbb H_*^n:=\mathbb H^n\backslash \{0\}$ and the Hamiltonian action of $\mathrm{Sp}(1)$ on   $T^*\mathbb H_*^n$ is induced from the natural right action of $\mathrm{Sp}(1)$ on   $\mathbb H_*^n$. )

\end{abstract}

\maketitle
 
\section {Introduction}
The Kepler problem is a textbook example of super integrable models. Its hamiltonian is invariant under the Lie group $\mathrm {SO}(4)$, larger than the manifest symmetry group $\mathrm{SO}(3)$. A remarkable fact about the Kepler problem is that  
the real non-compact Lie algebra $\frak{so}(4,2)$ has a nontrivial Poisson realization on its phase space. This Poisson realization, more precisely its quantized form, was initially discovered by I. A. Malkin and V. I. Man’ko \cite{Malkin66} in 1966.  (For the prehistory of this important discovery about the Kepler problem, one may consult Footnote 2 in Ref. \cite{Barut67}.) In the literature, the real non-compact Lie algebra $\frak{so}(4,2)$ is referred to as the \emph{dynamical symmetry algebra} for the Kepler problem and its afore-mentioned Poisson realization is referred to as the \emph{classical dynamical symmetry} for the Kepler problem. 

The Kepler problem has magnetized versions, under the name of MICZ-Kepler problems. The work of A. Barut and G. Bornzin \cite{Barut71}, extends the study of dynamical symmetry to
these magnetized Kepler problems at the quantum level.  Later, the higher dimensional analogue of MICZ-Kepler problems, under the name of generalized MICZ-Kepler problems, were found\cite{meng2007} and the dynamical symmetry for these models was studied as well, at both the quantum level \cite{mengzhang} and the classical level \cite{meng2013b}.

About six years ago the second author \cite{meng2011,meng2013} discovered that the Kepler problem has a vast generalization based on simple euclidean Jordan algebra, for which the conformal  algebra of the Jordan algebra is the dynamical symmetry algebra. He also made the following observation \cite{meng2014'}: for a generalized Kepler problem, its hamiltonian and its Laplace-Runge-Lenz vector can all be derived from its dynamical symmetry.

Recently, we exhibited in Ref. \cite{BM2015} the classical dynamical symmetry for the $\mathrm U(1)$-Kepler problems (see Ref. \cite{meng2014a} ). These Kepler-type problems are naturally associated with the euclidean Jordan algebras of complex hermitian matrices. In this article we shall exhibit the classical dynamical symmetry for the $\mathrm {Sp}(1)$-Kepler problems \cite{meng2014b}, i.e. the quaternionic analogues of the $\mathrm U(1)$-Kepler problems. As a result, we obtain the Laplace-Runge-Lenz vector for each $\mathrm {Sp}(1)$-Kepler problem as well as a formulae of expressing the total energy in terms of the angular momentum and the Laplace-Runge-Lenz vector.  For the convenience of readers, we end this introduction with 
\begin{center}
{\bf A list of symbols} 
\end{center}
\begin{eqnarray}
\begin{array}{ll}
\mathbb H & \quad\text{the set of quaternions}\cr
\mathrm H_n(\mathbb H)&\quad\text{the Jordan algebra of quaternionic hermitian matrices of order $n$}\cr
\mathpzc{Im}\,\mathbb H  & \quad\text{the set of imaginary quaternions}\cr
\mathrm i, \mathrm j, \mathrm k & \quad \text{the standard orthonormal basis for $\mathpzc{Im}\,\mathbb H $ such that $\mathrm i\mathrm j=\mathrm k$}\cr
\bar q &\quad\text{the quaternionic conjugate of quaternion $q$, e.g., $\bar{\mathrm k}=-{\mathrm k}$} \cr
\mathpzc{Re}\, q & \quad {1\over 2}(q+\bar q)\cr
\mathpzc{Im}\, q & \quad {1\over 2}(q-\bar q)\cr
\lrcorner & \quad\mbox{the interior product of vectors with forms}\cr
\wedge & \quad\mbox{the wedge product of forms}\cr
{\mathrm d} & \quad\mbox{the exterior derivative operator}\cr
\pi_X:\, T^*X\to X &\quad\mbox{the cotangent bundle projection}\cr
G & \quad\mbox{a compact connected Lie group} \cr
\mathfrak g, \mathfrak g^* & \quad\mbox{the Lie algebra of $G$ and its dual} \cr
\xi & \quad\mbox{an element in $\mathfrak g$}\cr
\langle\, ,\,\rangle & \quad\mbox{either the paring of vectors with co-vectors or inner product}\cr
 \langle\, \mid\,\rangle & \quad\mbox{the inner product on the Jordan algebra $\mathrm H_n(\mathbb H)$}\cr
\mathrm{Ad}_a & \quad\mbox{the adjoint action of $a\in G$ on $\mathfrak g$}\cr
P\to X & \quad\mbox{a principal $G$-bundle} \cr
\Theta & \quad\mbox{a $\mathfrak g$-valued differential one-form on $P$ that}\cr 
&\quad\mbox{defines a principal connection on $P\to X$} \cr
R_a & \quad\mbox{the right action on $P$ by $a\in G$}\cr
X_\xi & \quad\mbox{the vector field on $P$ which represents the}\cr
&\quad\mbox{infinitesimal right action on $P$ by $\xi\in \mathfrak g$}\cr
F & \quad\mbox{a hamiltonian $G$-space}\cr
\mathcal F:=P\times_G F& \quad\mbox{the quotient of $P\times F$ by the action of $G$}\cr
\Phi: F\to \mathfrak g^* & \quad\mbox{the $G$-equivariant moment map}\cr
\mathcal F^\sharp &\quad \text{Sternberg phase space}\cr
\mathcal W &\quad \text{Weinstein's universal phase space}\cr
\end{array}\nonumber
\end{eqnarray}

\section{The dynamical symmetry algebra}
The dynamic symmetry for the $\mathrm{Sp}(1)$-Kepler model at level $n$ with magnetic charge $\mu$ that we shall exhibit is a Poisson realization of the dynamic symmetry algebra $\frak{so}^*(4n)$ on its phase space. Note that $\frak{so}^*(4n)$,  being the conformal algebra of the simple euclidean Jordan algebra $\mathrm H_n(\mathbb H)$ of quaternionic hermitian matrices of order $n$, can be understood naturally in the language of Jordan algebra \cite{PJordan33}. The details are given in the next two paragraphs (see \cite{FK94}  for more details).

For each $u\in V:= \mathrm H_n(\mathbb H)$, we use $L_u$ to denote the Jordan multiplication by $u$, and for each $u, v\in V$, we let $S_{uv}=[L_u, L_v]+L_{uv}$ where $[L_u, L_v]$ stands for the commutator: $L_uL_v-L_vL_u$ and $uv$ in $L_{uv}$ means $L_u (v)$, i.e., the symmetrized matrix product of $u$ with $v$. We use $\{uvw\}$ to denote $S_{uv}(w)$. Then we have
\[
[S_{uv}, S_{zw}]=S_{\{uvz\}w}- S_{z\{vuw\}}\quad \mbox{for any $u, v, z, w$ in $V$. }
\] So these $S_{uv}$ span a real Lie algebra. This Lie algebra is denoted by $\frak{str}$, and is referred to as the structure algebra of $V$. In fact $\frak{str}=\frak{su}^*(2n)\oplus \mathbb R$ where the center $\mathbb R$ is generated by $L_e$ --- the Jordan multiplication by the Jordan identity element $e$. 

The conformal algebra $\frak{co}$ is an extension of the structure algebra $\frak{str}$. As a real vector space we have
\[
\frak{co}=V\oplus \frak{str}\oplus V^*.
\]
An element $z$ in $V$, rewritten as $X_z$, behaves like a vector:
\[
[S_{uv}, X_z]=X_{\{uvz\}}
\]
and an element in $V^*$ behaves like a co-vector. Via the inner product on $V$, we can identify this element in $V^*$ with an element $w$ in $W$, which is rewritten as $Y_w$, then
\[
[S_{uv}, Y_w]=-Y_{\{vuw\}}.
\]
The remaining commutation relations are
\[
[X_u, Y_v]=0, \quad [Y_u, Y_v]=0, \quad [X_u, Y_v]=-2S_{uv}\quad \mbox{for any $u, v$ in $V$. }
\]
One can verify that, indeed, $\frak{co} = \frak{so}^*(4n)$.

\section{The Phase Space}
When the magnetic charge is not zero, the phase space, being a Sternberg phase space \cite{Sternberg77}, is a bit involved, so the Poisson realization of the dynamic symmetry algebra on the the phase space is a bit complicated and the verification of various Poisson relations becomes quite tedious, as evidenced already in simpler models such as the $\mathrm U(1)$-Kepler problems \cite{BM2015}.  

To circumvent this complication, we resort to an insight of A. Weinstein into the Sternberg phase space.  As we shall see the Sternberg phase spaces form a bundle of symplectic manifolds over the affine space of principal connections.  Since its base space is contractible, this fiber bundle must be topologically trivial. Indeed,   A. Weinstein observed that \cite{Weinstein78} this bundle of symplectic manifolds is canonically isomorphic to a product bundle whose fiber is a fixed symplectic manifold, i.e., Weinstein's \emph{universal phase space}. 
\subsection{Review of the work by S. Sternberg and A. Weinstein}
The goal of this subsection is to review the work by S. Sternberg \cite{Sternberg77} and A. Weinstein \cite{Weinstein78}. Let us start with the setup for Sternberg phase space: 
\begin{enumerate}[(i)]
\item A compact Lie group $G$ and a principal $G$-bundle $P\to X$ with a principal connection form $\Theta$, 
\item A Hamiltonian $G$-space $F$ with symplectic form $\Omega$ and a $G$-equivariant  moment map $\Phi$: $F\to \frak g^*$. Here $\frak g$ is the Lie algebra of the Lie group $G$. 
\end{enumerate}
Note that, a co-adjoint orbit of $G$ with the Kirillov-Kostant-Souriau symplectic form is a typical example of Hamiltonian $G$-space.

For the convenience of readers, let us recall that a principal connection $\Theta$ on the principal $G$-bundle $P\to X$ is a $\frak g$-valued differential one-form on $P$ satisfying the following two conditions: 
\begin{center}
1) ${R_{a^{-1}}}^*\,\Theta ={\mathrm {Ad}}_a\Theta$ for any $a\in G$, \quad\quad 2) $\Theta(X_\xi)=\xi$ for any $\xi\in \mathfrak g$.
\end{center} Here $a\in G$, $R_{a^{-1}}$ is the right multiplication of $a^{-1}$ on $P$, ${\mathrm {Ad}}_a$ is the adjoint action of $a$ on $\frak g$, and $X_\xi$ is the vector field on $P$ which represents
the induced action of $\xi\in \frak g$ on $P$, i.e., for any $f\in C^\infty(P)$, we have the Lie derivative
\begin{eqnarray}
{\mathcal L} _{X_\xi}f \left|_{p} =\frac{\mathrm d}{\mathrm dt} \right|_{t=0}f(p\cdot \exp (t\xi)).
\end{eqnarray}
It is easy to see that $[\mathcal L _{X_\xi}, \mathcal L _{X_\eta}] =  \mathcal L _{X_{[\xi, \eta]}}$, or equivalently  $[X_\xi, X_\eta] = X_{[\xi, \eta]}$.

It is a tautology that a smooth (right) $G$-action on a smooth manifold $P$ yields a hamiltonian $G$-action on the symplectic manifold $T^*P$ with the $G$-equivariant  moment map $\rho$: $T^*P\to \frak g^*$. Indeed, if we use $\pi_P$ to denote the bundle projection $T^*P\to P$, then $\rho$ is defined via equation
\[
\langle \rho(z), \xi\rangle =\langle z, X_\xi (\pi_P(z))\rangle.
\]
Another way to see it is this: for each point $p\in P$,  the dual of linear map 
\begin{eqnarray}
\frak g &\to& T_pP\cr
\xi &\mapsto & X_\xi(p)
\end{eqnarray} is a linear map from $T_p^*P$ to $\frak g^*$. Assembling these dual maps together, we get the map $\rho$: $T^*P\to \frak g^*$. 

The $G$-equivariance of $\rho$ is reflected by the fact that $(R_g)_*(X_\xi)=X_{\mathrm{Ad}_{g^{-1}}(\xi)}$ and the fact that $\mathrm{Ad}_g^*(\alpha)=\alpha\circ \mathrm{Ad}_{g^{-1}}$ for any $\alpha\in \frak g^*$. The map $\rho$ is a moment map, i.e.,  $\mathrm d \langle \rho, \xi\rangle=-\hat X_\xi\;\lrcorner\; \omega_P$. Here $\omega_P$ is the canonical symplectic form on $T^*P$, and the vector field $\hat X_\xi$ on $T^*P$ is the cotangent lift of  the vector field $X_\xi$. In local coordinates, we have
\[
z =p_i(z)\,\mathrm d x^i|_{\pi_P(z)}, \quad \omega_P=\mathrm d p_i\wedge \mathrm d x^i, \quad X_\xi= X_\xi^i\frac{\partial}{\partial x^i}, \quad \hat X_\xi= X_\xi^i\frac{\partial}{\partial x^i} -p_i\frac{\partial X_\xi^i}{\partial x^j}\frac{\partial }{\partial p_j}
\] and $\langle \rho, \xi\rangle =  p_iX^i_\xi$. Here $x^i$ is a system of local coordinates on $P$. 

Combining $\rho$ and $\Phi$, we obtain a $G$-equivariant moment map
\[
\psi:\quad T^*P\times F\to \frak g^*.
\] which maps $(x, y)$ to $-\rho(x)+\Phi(y)$.

Since $P\to X$ is a principal $G$-bundle,  $\rho|_{T^*_pP}$: $T_p^*P\to \frak g^*$ is a diffeomorphism for each $p\in P$, then $\rho$: $T^*P\to \frak g^*$ is a submersion. Consequently $\psi$ is a submersion as well. In particular, this means that 
$\psi^{-1}(0)$ is a submanifold of $T^*P\times F$. Since the isotropic group of $0\in \frak g^*$, being the Lie group $G$, is compact, its free action on $\psi^{-1}(0)$ is proper, Theorem 1 in Ref. \cite{Marsden74} applies, so there is a unique symplectic structure $\omega$ on $\psi^{-1}(0)/G$ such that 
$\pi^*\omega=\iota^*(\omega_P+\Omega)$ where $\omega_P$ is the tautological symplectic form on $P$, $\pi$ is the projection and $\iota$ is the inclusion:
\begin{eqnarray}
\begin{tikzcd}[column sep=small]
\psi^{-1}(0) \arrow{d}[swap]{\pi} \arrow{r}{\iota} &  T^*P\times F\\
\psi^{-1}(0)/G& 
\end{tikzcd}\nonumber
\end{eqnarray} 
In the literature this reduced phase space $(\psi^{-1}(0)/G, \omega)$ is called the Weinstein's \emph{universal phase space} $\mathcal W$. Note that, no connection is required for the existence of this universal phase space. 

To understand the meaning of the word ``universal", let us  suppose that a principal connection $\Theta$ on $P\to X$ is given. For each point $p\in P$, let $x$ be the image of $p$ under the bundle projection map, then the equivariant horizontal lifting of tangent vectors on $X$ (provided by $\Theta$) defines a linear map $T_xX\to T_pP$. By assembling the dual of these linear maps together, we arrive at the commutative square
\begin{eqnarray}
\begin{tikzcd}[column sep=small]
T^*P \arrow{d}[swap]{\pi_P}\arrow{r}  &  T^*X\arrow{d}{\pi_X} \\
P\arrow{r} & X
\end{tikzcd}\nonumber
\end{eqnarray} where the top arrow, fiber-wise speaking, is the dual of the horizontal lifting of tangent vectors of $X$ to tangent vectors of $P$. Let $\tilde P$ be the pullback of 
\begin{eqnarray}
\begin{tikzcd}[column sep=small]
&  T^*X\arrow{d}{\pi_X} \\
P\arrow{r} & X
\end{tikzcd}\nonumber
\end{eqnarray}
then we have a $G$-equivariant map $T^*P\to \tilde P$, hence, by taking its product with the identity map on $F$, we obtain a $G$-equivariant map
\[
\alpha_\Theta: T^*P\times F\to \tilde P\times F.
\]
Next, Weinstein observed that the restriction of $\alpha_\Theta$ to $\psi^{-1}(0)$ is a diffeomorphism; then, passing to the quotient by the action of $G$, one obtains a diffeomorphism
\[
\overline \alpha_\Theta: \quad \psi^{-1}(0)/G\to \tilde P\times_GF. 
\] 
So there is a unique symplectic structure $\omega_\Theta$ on $\mathcal F^\sharp:= \tilde P\times_GF$ such that ${\overline\alpha_\Theta}^*(\omega_\Theta)=\omega$. Then $(\mathcal F^\sharp, \omega_\Theta)$ is a symplectic manifold,  which is precisely the Sternberg phase space of $\Theta$ described in Ref. \cite{Sternberg77}. In other words, the Sternberg phase spaces form a bundle of symplectic manifolds over the space of principal connections, and this bundle is canonically isomorphic to the product bundle whose fiber is Weinstein's universal phase space
$\mathcal W$. 

\subsection{Sternberg phase space for $\mathrm {Sp}(1)$-Kepler problems}
The classical $\mathrm {Sp}(1)$-Kepler problems (or models) are indexed by integer (called \emph{level}) $n\ge 2$ and real number (called \emph{magnetic charge}) $\mu\ge 0$. For the model at level $n$ and magnetic charge $\mu$, its phase space is the Sternberg phase space $\mathcal F_\mu^\sharp$ with the following data \cite{meng2014b}:
\begin{enumerate}[(i)]
\item The compact Lie group $G$ is $\mathrm{Sp}(1)$(i.e. $\mathrm{SU}(2)$) and the principal
$G$-bundle is 
\begin{eqnarray}
\mathbb H^n_*&\to & \mathcal C_1\cr
Z&\mapsto &n ZZ^\dag 
\end{eqnarray} and the principal connection form is
\begin{eqnarray}\label{principalConnection}
\Theta = {\mathpzc{Im}(\bar Z\cdot \mathrm{d}Z)\over |Z|^2}.
\end{eqnarray}
Here $\mathcal C_1$ is the rank-one Kepler cone for the simple Euclidean Jordan algebra $\mathrm H_n(\mathbb H)$ of quarternonic hermitian matrices of order $n$. As a submanifold of the Euclidean space $\mathrm H_n(\mathbb H)$,  $\mathcal C_1$ consists of all rank one semi-positive definite elements of $\mathrm H_n(\mathbb H)$. However,  $\mathcal C_1$ is not a Riemannian submanifold of the Euclidean space $\mathrm H_n(\mathbb H)$ because the Riemannian metric on $\mathcal C_1$, called the Kepler metric, does not come from the Euclidean metric via restriction, see Ref. \cite{meng2014b} for the details.

\item For simplicity, we shall identify $\mathfrak g^*$ with $\mathfrak g:= \mathpzc{Im}\, \mathbb H $ via the standard invariant inner product $\langle, \rangle$, i.e. the one such that the imaginary units $\mathrm i$, $\mathrm j$ and $\mathrm k$ form an orthonormal basis for $\mathfrak g$. Then the Hamiltonian $G$-space is
\begin{eqnarray}
F:=\{\xi\in  \mathpzc{Im}\, \mathbb H \mid \xi\bar \xi =\mu^2\}
\end{eqnarray}
whose symplectic form $\Omega_\mu$, being the Kirillov-Kostant-Souriau symplectic form, is given by the formulae 
\begin{eqnarray} 
\Omega_\mu ={\langle\xi, \mathrm{d}\xi\wedge \mathrm{d}\xi \rangle\over 2|\xi|^2}. 
\end{eqnarray} 
Note that, if we write $\xi = \xi^1\mathrm i+ \xi^2\mathrm j+\xi^3\mathrm k$, then this symplectic form yields the following basic Poisson relations on $F$:
\begin{eqnarray}\label{poissonrelation}
\{\xi^1, \xi^2\}=\xi^3, \quad \{\xi^2, \xi^3\}=\xi^1, \quad\{\xi^3, \xi^1\}=\xi^2.
\end{eqnarray}
(Note: In our convention, $\mathrm i\mathrm j=\mathrm k$.)

\item
The $G$-equivariant moment map is
\begin{eqnarray}\label{phi}
\Phi: F&\to& \mathpzc{Im}\, \mathbb H \cr
\xi &\mapsto & 2\xi 
\end{eqnarray}
(Here the $G$-action on both $F$ and $ \mathpzc{Im}\, \mathbb H$ is the adjoint action.) Indeed, $\Phi$ is (obviously) $G$-equivariant; moreover, with a seemingly odd factor of $2$ included in Eq. \eqref{phi}, one can verify that
\[
\mathrm d \langle \Phi, \eta\rangle = X_\eta\;\lrcorner\; \Omega_\mu.
\]  
Here $X_\eta$ is the vector field on $F$ that represents the adjoint action on $F$ by the Lie algebra element $\eta$, so $X_\eta(\xi)=(\xi, [\eta, \xi])\in T_\xi F$.

\end{enumerate}

\subsection{Weinstein's Universal Phase Space for $\mathrm {Sp}(1)$-Kepler problems}
The right action of $\mathrm{Sp}(1)$ on $\mathbb H_*^n$ is the map that sends $(Z, q)\in \mathbb H_*^n\times  \mathrm{Sp}(1)$ to $Z\,[q]$ --- the matrix multiplication of the column matrix $Z$ with the $1\times 1$-matrix $[q]$.  The action induces a hamiltonian $\mathrm{Sp}(1)$-action on $T^*\mathbb H_*^n$ with a tautological moment map.  

The canonical trivialization of $T{\mathbb H}_*^n$ yields two $\mathbb H$-valued function on  
$T{\mathbb H}_*^n$, i.e. the position vector $Z$ and velocity vector $W$. 
Via the standard inner product on $\mathbb H^n$: \[
 \langle U, V\rangle =\mathpzc{Re}(U^\dag V),
\] one can identify the total cotangent space $T^*{\mathbb H}^n_*$ with the total tangent space $T{\mathbb H}^n_*$, then $T{\mathbb H}^n_*$ become a Poisson manifold with the following basic Poisson relation: for any $U, V\in \mathbb H^n$, 
\[
\{\langle U, Z\rangle, \langle V, W\rangle\}= \langle U, V\rangle, \quad \{\langle U, Z\rangle, \langle V, Z\rangle\}=\{\langle U, W\rangle, \langle V, W\rangle\} =0.
\]

With the aforementioned identification of $T^*\mathbb H_*^n$ with $T\mathbb H_*^n$ and $\mathfrak g^*$ with $\mathfrak g$ as well, one can check that the moment map $\rho$ is identified with the map from $T{\mathbb H}_*^n$ to $\frak g$ that maps $(Z, W)$ to $-\mathpzc{Im} (W^\dag Z)$. Therefore, the moment map $\psi$: $T\mathbb H_*^n\times F\to \mathfrak g$ is
\begin{eqnarray}
\psi(Z, W, \xi)=\mathpzc{Im} (W^\dag Z)+2\xi.
\end{eqnarray}
Note that the action of $g$ maps $(Z, W, \xi)$ to $(Z\cdot g^{-1}, W\cdot g^{-1}, g\xi g^{-1})$.

The map $\psi$ has a natural extension to $\tilde \psi$: $T\mathbb H_*^n\times  \mathfrak g\to \mathfrak g$ which is defined by the same formulae: 
\(
\tilde \psi(Z, W, \xi)=\mathpzc{Im} (W^\dag Z)+2\xi.
\) Then $\tilde \psi^{-1}(0)$ is the graph of the map $\xi =-\frac{1}{2} \mathpzc{Im} (W^\dag Z)$.  This map is clearly $G$-equivariant. Moreover, it is a Poisson map. Indeed, for example, \begin{eqnarray}
\{ \langle \mathrm i, W^\dag Z\rangle,  \langle \mathrm j, W^\dag Z\rangle\} &=& \{ \langle W\mathrm i, Z\rangle,  \langle  W\mathrm j, Z\rangle\} \cr
&=&  \{ \langle W\mathrm i, \contraction{}{Z}{\rangle,  \langle }{W}Z\rangle,  \langle  W\mathrm j, Z\rangle\} -<\mathrm i \leftrightarrow \mathrm j>\cr
&=& - \{ \langle W\mathrm i, \contraction{}{Z}{\rangle,  \langle }{W}Z\rangle,  \langle  W, Z\mathrm j\rangle\} -<\mathrm i \leftrightarrow \mathrm j> \cr
&=& -  \langle W\mathrm i, Z\mathrm j\rangle-<\mathrm i \leftrightarrow \mathrm j>\cr
&=& -  2\langle \mathrm k, W^\dag Z\rangle.\nonumber
\end{eqnarray} So, in view of the fact that $\xi^1=-\frac{1}{2} \langle \mathrm i, W^\dag Z\rangle$,
$\xi^2=-\frac{1}{2} \langle \mathrm j, W^\dag Z\rangle$, $\xi^3=-\frac{1}{2} \langle \mathrm k, W^\dag Z\rangle$, we have $\{\xi^1, \xi^2\}=\xi^3$, cf. Eq. \eqref{poissonrelation}.

Therefore, the $\mathrm{Sp}(1)$-equivariant projection map $T\mathbb H_*^n\times  \mathfrak g\to T\mathbb H_*^n$, when restricted to $\tilde \psi^{-1}(0)$, yields a $\mathrm{Sp}(1)$-equivariant Poisson isomorphism of $\tilde \psi^{-1}(0)$ with $T\mathbb H_*^n$. Since $\psi^{-1}(0)$ is a $\mathrm{Sp}(1)$-equivariant submanifold of $\tilde\psi^{-1}(0)$, the universal phase space 
\[
\mathcal W_\mu:=\psi^{-1}(0)/\mathrm{Sp}(1) =\left\{\mathrm{Sp}(1)\cdot (Z, W, \xi)\mid \xi =-\frac{1}{2} \mathpzc{Im} (W^\dag Z), |\xi|=\mu \right \}
\] is naturally identified with
a submanifold of $T\mathbb H_*^n/\mathrm{Sp}(1)$.  Indeed, as a symplectic manifold,  $\mathcal W_\mu$ is naturally identified with the Poisson leave
\[
\left\{\mathrm{Sp}(1)\cdot (Z, W) \mid |\mathpzc{Im} (W^\dag Z)|=2\mu \right \}
\] of the Poisson manifold $T\mathbb H_*^n/\mathrm{Sp}(1)$. In other words, in view of the fact that the Sternberg phase space $\mathcal F_\mu^\sharp$ can be identified with the Weinstein's universal phase space $\mathcal W_\mu$, \emph{the totality of Sternberg phase space can be identified with the Poisson manifold $T\mathbb H_*^n/\mathrm{Sp}(1)$}.

Part (ii) of Lemma \ref{key lemma} in the next section says that \emph{the conformal algebra of $\mathrm H_n(\mathbb H)$ has a Poisson realization on the Poisson manifold $T\mathbb H_*^n/\mathrm{Sp}(1)$}, hence on both the universal phase space $\mathcal W_\mu$ and the Sternberg phase space $\mathcal F_\mu^\sharp$.

\section{Dynamical Symmetry}
Let us fix a $\mathrm{Sp}(1)$-Kepler problem, say at level $n$ and with magnetic charge $\mu$.
Its phase space, being the Sternberg phase space $\mathcal F_\mu^\sharp$, fibers over $T^*\mathcal C_1$. Via the metric on the Euclidean space $V:=\mathrm H_n(\mathbb H)$, $T^*\mathcal C_1$ can be identified with $T\mathcal C_1$, so $\mathcal F_\mu^\sharp$ fibers over $T\mathcal C_1$ as well.  Let functions $x$ and $\pi$ (taking value in vector space $V$) be defined via diagram
\begin{eqnarray}
\begin{tikzcd}[column sep=small]
 &T\mathcal C_1 \arrow{d}[description]{\iota}\arrow{ldd}[swap]{x}\arrow{rdd}{\pi} & \\
&TV\arrow{ld}[description]{\tau_V}\arrow{rd} [description]{t}& \\
V&&V
\end{tikzcd}\nonumber
\end{eqnarray} where $\iota$ is the inclusion map and $\tau_V$ is the tangent bundle projection and $t$ is the natural trivialization map of the tangent bundle of the affine space $V$. The pullback of $x$ and $\pi$ under the bundle map $\mathcal F_\mu^\sharp\to T\mathcal C_1$ shall still be denoted by $x$ and $\pi$. Recall that the inner product on the Jordan algebra $\mathrm H_n(\mathbb H)$ is denoted by $\langle \,\mid\, \rangle$, so if $u\in \mathrm H_n(\mathbb H)$, then $\langle x\mid u\rangle$ is a real function on $\mathcal F_\mu^\sharp$.

\begin{Thm}\label{main theorem}
For the conformal algebra of the Jordan algebra $\mathrm H_n(\mathbb H)$, there is a unique Poisson realization on the Sternberg phase space $\mathcal F^\sharp_\mu$ such that $Y_u$ is realized as function $\mathcal Y_u:=\langle x\mid u\rangle$ and $X_e$ is realized as function $\mathcal X_e:=\langle x\mid \pi^2\rangle +\frac{\mu^2}{ \langle e\mid x\rangle}$. Moreover, if $(e_\alpha)$ is an orthonormal basis for $\mathrm H_n(\mathbb H)$ and $\mathcal L_u$ represents $L_u$ in this Poisson realization, we have the following primary quadratic relation:
\begin{eqnarray} 
{2\over n}\sum_\alpha \mathcal L_{e_\alpha}^2 =\mathcal  L_e^2 + \mathcal  X_e  \mathcal Y_e -\mu^2.
\end{eqnarray} 
 \end{Thm}
The uniqueness of this Poisson realization (if it exists) is clear, that is because $X_e$ and $Y_u$ generate the conformal algebra.

As for the existence of this Poisson realization, the direct verification is very complicated. Since Sternberg phase space can be identified with Weinstein's universal phase space via diffeomorphism 
\begin{eqnarray}
\overline \alpha_\Theta: \quad \mathcal W_\mu:=\psi^{-1}(0)/\mathrm{Sp}(1)\to \mathcal F_\mu^\sharp:=\tilde P\times_{\mathrm{Sp}(1)} F, 
\end{eqnarray}
one just needs to verify the existence of the corresponding Poisson realization on Weinstein's universal phase space $\mathcal W_\mu$, a task which turns to be much simpler. 

\begin{Lem}\label{key lemma}
\begin{enumerate}[(i). ]
\item Let $\mathrm{Sp}(1)\cdot (Z, W, \xi)$ be an element of $\mathcal W_\mu$, and let $\mathcal Y_u$ and  $\mathcal X_e$ be the functions defined in Theorem \ref{main theorem} . Then
\[
\mathcal Y_u\circ \overline \alpha_\Theta (\mathrm{Sp}(1)\cdot (Z, W, \xi)) = \langle Z, uZ\rangle, \quad \mathcal X_e\circ \overline \alpha_\Theta (\mathrm{Sp}(1)\cdot (Z, W, \xi)) =\frac{1}{4}|W|^2.
\]

\item For any vectors $u$, $v$ in $V:=\mathrm{H}_n(\mathbb H)$,  define $\mathrm{Sp}(1)$-invariant functions 
\begin{eqnarray} 
\left\{
\begin{array}{rcl}
\mathscr X_u &:=& \frac{1}{4} \langle W, uW\rangle \\[2mm]
\mathscr Y_v &:=& \langle Z, vZ\rangle\\[2mm]
\mathscr S_{uv} &:=& \frac{1}{2} \langle W, (u\cdot v) Z \rangle
\end{array}\right.
\end{eqnarray}
on $T{\mathbb H}_*^n$. Here $u\cdot v$ means the matrix multiplication of $u$ with $v$.  Then, for any vectors $u$, $v$, $z$, $w$ in $V$,  the following Poisson bracket relations hold:
\begin{eqnarray}
\left\{
\begin{matrix}
\{\mathscr X_u, \mathscr X_v\} =0, \quad \{\mathscr Y_u, \mathscr Y_v\}=0, \quad \{\mathscr X_u,
\mathscr Y_v\} = -2\mathscr S_{uv},\cr\\ 
\{\mathscr S_{uv}, \mathscr X_z\}=\mathscr X_{\{uvz\}}, \quad \{\mathscr S_{uv}, \mathscr Y_z\} = -\mathscr Y_{\{vuz\}},\cr\\
\{\mathscr S_{uv}, \mathscr S_{zw}\} = \mathscr S_{\{uvz\}w}-\mathscr S_{z\{vuw\}}.
\end{matrix}\right.\nonumber
\end{eqnarray}
Consequently, we have a Poisson realization on the Poisson manifold  $T^*{\mathbb H}_*^n/\mathrm{Sp}(1)$ for the conformal algebra of the Jordan algebra $\mathrm H_n(\mathbb H)$.
\item
Let $e$ be the identity element of $V$, $(e_\alpha)$ be an orthonormal basis for $V$, and
$\mathscr L_u=\mathscr S_{eu}$ for any $u\in V$.  Then
\begin{eqnarray} 
{2\over n}\sum_\alpha \mathscr L_{e_\alpha}^2=\mathscr L_e^2 + \mathscr  X_e  \mathscr Y_e  -\mu^2.
\end{eqnarray}

\end{enumerate}
\end{Lem}

\begin{proof}
\begin{enumerate}[(i). ]
\item To understand the map $\overline \alpha_\Theta$, we need to figure out the horizontal lift induced from the connection $\Theta$. For $Z\in \mathbb H_*^n$, we let $x=nZZ^\dag$. Suppose that $(x, \dot x)$ is a tangent vector of $\mathcal C_1$ at $x$, and $(Z, \dot Z)$ is its horizontal lift to point $Z$ in $\mathbb H_*^n$. Then, in view of Eq. \eqref{principalConnection}, we have equations
\[
n(\dot Z Z^\dag + Z\dot Z^\dag) =\dot x, \quad \mathpzc{Im} (Z^\dag \dot Z) =0.
\] By solving these equations jointly, we obtain 
\[
\dot Z = \frac{1}{n|Z|^2}(\dot x Z  -\frac{\tr \dot x}{2}Z).
\] 

Suppose that the $\Theta$-induced map $T^*\mathbb H_*^n\to T^*\mathcal C_1$ maps
$(Z, \langle W, \, \rangle)$ to $(x, \langle \pi\mid\;\rangle)$, then $x=nZZ^\dag$ and
\begin{eqnarray}\label{hlift}
\langle \pi\mid \dot x\rangle = \frac{1}{n|Z|^2}\left\langle W,  \dot x Z  -\frac{\tr \dot x}{2}Z\right\rangle.
\end{eqnarray}
In particular, in view of the fact that $(x, ux)\in T_x\mathcal C_1$, we have 
\begin{eqnarray}\label{simpleformula}
\langle \pi\mid ux\rangle &= &  \frac{1}{n|Z|^2}\left\langle W,  (ux) Z  -\frac{\tr (ux)}{2}Z\right\rangle\quad \text{$ux$ is the Jordan product of $u$ with $x$}\cr
&=& \frac{1}{2|Z|^2}\left\langle W,  uZZ^\dag Z+ZZ^\dag uZ  - \mathpzc{Re}\,\tr (uZZ^\dag)Z\right\rangle\cr
&=&\frac{1}{2}\langle W, uZ\rangle.
\end{eqnarray}
Then 
\begin{eqnarray}
\overline \alpha_\Theta (G\cdot (Z, W, \xi))=G\cdot (Z, nZZ^\dag, \pi, \xi).
\end{eqnarray}
Consequently
\[
 \mathcal Y_u\circ \overline \alpha_\Theta (G\cdot (Z, W, \xi))=\langle x\mid u\rangle =\langle nZZ^\dag\mid u\rangle =\mathpzc{Re}\, \tr (ZZ^\dag u)=\langle Z, uZ\rangle
\]
and
\begin{eqnarray}
 \mathcal X_e\circ \overline \alpha_\Theta (G\cdot (Z, W, \xi)) &= & \langle x\mid \pi^2\rangle +\frac{\mu^2}{ \langle e\mid x\rangle}=\langle \pi\mid \pi x\rangle +\frac{n\mu^2}{ \tr  x}\cr
 &=& \frac{1}{2}\langle W, \pi Z\rangle+\frac{\mu^2}{|Z|^2}\quad{\text{using Eq.}\, \eqref{simpleformula}}\cr
  &=& \frac{n}{2}\langle \pi\mid (ZW^\dag)_+\rangle+\frac{\mu^2}{|Z|^2}. \nonumber
\end{eqnarray}
Here $(ZW^\dag)_+=\frac{1}{2}(ZW^\dag+WZ^\dag)$. One can check that $(x, (ZW^\dag)_+)$ is a tangent vector of $\mathcal C_1$ at $x$. Therefore, in view of Eq. \eqref{hlift}, we have
\begin{eqnarray}
 \mathcal X_e\circ \overline \alpha_\Theta (G\cdot (Z, W, \xi)) 
  &=&\frac{1}{2|Z|^2}\left\langle W,  (ZW^\dag)_+ Z  -\frac{\tr (ZW^\dag)_+}{2}Z\right\rangle+\frac{\mu^2}{|Z|^2}\cr
 &=&\frac{1}{2|Z|^2}\left\langle W,  (ZW^\dag)_+ Z  -\frac{\langle W, Z\rangle}{2}Z\right\rangle+\frac{\mu^2}{|Z|^2} \cr
  &=&\frac{1}{2|Z|^2}\left\langle W,  (ZW^\dag)_+ Z\right\rangle  - \frac{1}{4|Z|^2}\langle W, Z\rangle^2+\frac{\mu^2}{|Z|^2}\cr
&=&\frac{1}{4|Z|^2}\left(|W|^2|Z|^2+ \mathpzc{Re}\, (W^\dag Z)^2 -\langle W, Z\rangle^2\right)+\frac{\mu^2}{|Z|^2}\cr
&=&\frac{1}{4|Z|^2}\left(|W|^2|Z|^2+ (\mathpzc{Im} (W^\dag Z) )^2\right)+\frac{\mu^2}{|Z|^2}\cr
&=&\frac{1}{4|Z|^2}\left(|W|^2|Z|^2 - |\mathpzc{Im} (W^\dag Z)|^2\right)+\frac{\mu^2}{|Z|^2}\cr
&=&\frac{1}{4}|W|^2.\nonumber
  \end{eqnarray}
  
\item 
It is clear that $\{\mathscr X_u, \mathscr X_v\} =0$ and $\{\mathscr Y_u, \mathscr Y_v\}=0$. Next, we have
\begin{eqnarray}
\{\mathscr X_u, \mathscr Y_v\}& =&\frac{1}{4} \left \{  \langle W, uW\rangle,  \langle Z, vZ\rangle\right \}\cr
&=&\left \{\langle \contraction{}{W}{, uW\rangle,  \langle} {Z}W, uW\rangle,  \langle Z, vZ\rangle\right \} = -\langle uW, vZ\rangle   \cr
&=& -2\mathscr S_{uv}\nonumber
\end{eqnarray}
and
\begin{eqnarray}
\{\mathscr S_{uv}, \mathscr S_{zw}\} & =&\frac{1}{4} \left \{ \langle W, (u\cdot v) Z \rangle,  \langle W, (z \cdot w) Z \rangle\right \}\cr
&=& \frac{1}{4}  \left \{ \langle  \contraction{}{W}{, (u\cdot v) Z \rangle,  \langle W, (z\cdot w) }{Z}W, (u\cdot v) Z \rangle,  \langle W, (z\cdot  w) Z \rangle\right \}   -<(u\cdot v)\leftrightarrow (z\cdot w)>\cr
&=&  -\frac{1}{4}\langle (z\cdot w)^\dag W, (u\cdot v)Z\rangle  + \frac{1}{4} \langle (u\cdot v)^\dag W, (z\cdot w)Z\rangle \cr
&=&\frac{1}{4}\langle W, [u\cdot v, z\cdot w]Z\rangle\cr
&=&  \mathscr S_{\{uvz\}w}-\mathscr S_{z\{vuw\}}\nonumber
\end{eqnarray}
because $\{uvz\}=\frac{1}{2}(u\cdot v\cdot z+z\cdot v\cdot u)$ and $\{vuw\}=\frac{1}{2}(v\cdot u\cdot w+w\cdot u\cdot v)$.

Thirdly,
\begin{eqnarray}
\{\mathscr S_{uv}, \mathscr X_{z}\} & =&\frac{1}{8} \left \{ \langle W, (u\cdot v) Z \rangle,  \langle W,  z W \rangle\right \} = \frac{1}{4} \left \{ \langle W, (u\cdot v) \contraction{}{Z}{\rangle,  \langle }{W}Z \rangle,  \langle W,  z W \rangle\right \} \cr
&=&  \frac{1}{4}\langle W, (u\cdot v\cdot z)W\rangle  =   \frac{1}{4}\langle W, (z\cdot v\cdot u)W\rangle  \cr
&=&\frac{1}{8}\langle W,(z\cdot v\cdot u+u\cdot v\cdot z)W\rangle\cr
&=&  \mathscr X_{\{uvz\}}.\nonumber
\end{eqnarray}

Finally, we have
\begin{eqnarray}
\{\mathscr S_{uv}, \mathscr Y_{z}\} & =&\frac{1}{2} \left \{ \langle W, (u\cdot v) Z \rangle,  \langle Z,  z Z \rangle\right \} =  \left \{ \langle  \contraction{}{W}{, (u\cdot v)Z \rangle,  \langle  }{Z}W, (u\cdot v)Z \rangle,  \langle Z,  z Z \rangle\right \} \cr
&=& - \langle Z, (z\cdot u\cdot v)Z  \rangle =  - \langle Z, (v\cdot u\cdot z)Z  \rangle\cr
&=&-\frac{1}{2}\langle Z, (z\cdot u\cdot v+v\cdot u\cdot z)Z\rangle\cr
&=&  -\mathscr Y_{\{vuz\}}.\nonumber
\end{eqnarray}

Since functions $\mathscr S_{uv}$, $\mathscr X_z$ and $\mathscr Y_w$ on $T{\mathbb H}_*^n$ (actually on $T^*{\mathbb H}_*^n$) are $\mathrm{Sp}(1)$-invariant, and the action of $\mathrm{Sp}(1)$ on $T^*{\mathbb H}_*^n$ is symplectic, we have a Poisson realization on the Poisson manifold  $T^*{\mathbb H}_*^n/\mathrm{Sp}(1)$ for the conformal algebra of the Jordan algebra $\mathrm H_n(\mathbb H)$.

\item Since $\mathscr L_u=\frac{1}{2}\langle W, uZ\rangle$, we have
\begin{eqnarray}
{2\over n}\sum_\alpha \mathscr L_{e_\alpha}^2& =&\frac{1}{2n} \sum_\alpha \langle W, e_\alpha Z\rangle^2 = \frac{n}{2} \sum_\alpha \langle (ZW^\dag)_+\mid e_\alpha \rangle^2 \cr
&= &  \frac{n}{2} \langle (ZW^\dag)_+\mid (ZW^\dag)_+ \rangle = \frac{1}{2} \tr \left ((ZW^\dag)_+\right )^2\cr
&=& \frac{1}{8} \tr \left (ZW^\dag+WZ^\dag\right )^2\cr
&=& \frac{1}{8} \tr (Z(W^\dag ZW^\dag)+(WZ^\dag W)Z^\dag+|W|^2 ZZ^\dag+|Z|^2WW^\dag)\cr
&=& \frac{1}{8}  ( (W^\dag Z)^2+(Z^\dag W)^2)+\frac{1}{4}|W|^2 |Z|^2\cr
&=& \frac{1}{16}  (W^\dag Z-Z^\dag W)^2+ \frac{1}{16}  (W^\dag Z+Z^\dag W)^2+\frac{1}{4}|W|^2 |Z|^2\cr
&=& -\frac{1}{4}  |\mathpzc{Im} (W^\dag Z)|^2+\frac{1}{4}\langle W, Z\rangle^2+\frac{1}{4}|W|^2 |Z|^2\cr
&=& -\mu^2 +\mathscr L_e^2+\mathscr X_e\mathscr Y_e.\nonumber
\end{eqnarray}
\end{enumerate}

\end{proof}
\underline{Proof of Theorem \ref{main theorem}}.  As we remarked early that the Poisson realization is obviously unique, if it exists. From part (ii) of the above lemma,  we know that there is a Poisson realization for $\frak{so}^*(4n)$ on the Poisson manifold  $T^*{\mathbb H}_*^n/\mathrm{Sp}(1)$, hence on each of its symplectic leave. Since a Sternberg phase space is symplectic equivalent to a symplectic leave of $T^*{\mathbb H}_*^n/\mathrm{Sp}(1)$, in view of part (i) of the above lemma, one obtains the first part of Theorem \ref{main theorem}.  The remaining part of   
Theorem \ref{main theorem} follows trivially from part (iii) of the above lemma.  

\begin{rmk}
The Poisson realization in Theorem \ref{main theorem} is called the dynamical symmetry for the $\mathrm{Sp}(1)$-Kepler problem at level $n$ and with magnetic charge $\mu$. 
\end{rmk}
\begin{rmk}The natural identification of Sternberg phase spaces with symplectic leaves of $T^*{\mathbb H}_*^n/\mathrm{Sp}(1)$ is a bijection, so the totality of Sternberg phase spaces is naturally equivalent to 
the Poisson manifold $T^*{\mathbb H}_*^n/\mathrm{Sp}(1)$. This really indicates that $\mathrm{Sp}(1)$ Kepler problems can be obtained via symplecic-reduction from a model whose phase space is $T^*{\mathbb H}_*^n$ and whose hamiltonian is $\mathrm{Sp}(1)$-invariant, namely the $n$th quaternionic conformal Kepler problem in section 6 of Ref. \cite{meng2014b}.

\end{rmk}

\begin{rmk} In view of Ref. \cite{meng2014'}, Theorem \ref{main theorem} implies that the corresponding $\mathrm{Sp}(1)$-Kepler problem is the Hamiltonian system with phase space $T\mathcal C_1$, Hamiltonian 
 \[
 H={1\over 2} {\mathcal X_e\over \mathcal Y_e} - {1\over \mathcal Y_e} \]
 and Laplace-Runge-Lenz vector
 \[
 \mathcal A_u={1\over 2}\left(\mathcal X_u - \mathcal Y_u {\mathcal X_e\over \mathcal Y_e}\right)+{\mathcal Y_u\over \mathcal Y_e}
 \] where $\mathcal X_u$ and $\mathcal Y_u$ are the functions that represent $X_u$ and $Y_u$ respectively in the Poisson realization in Theorem \ref{main theorem}. Indeed, a simple computation yields 
 \[
 H=\frac{1}{2}\frac{\langle x|\pi^2\rangle}{r} +{\mu^2\over 2r^2} -\frac{1}{r}, 
 \] i.e., the Hamiltonian in Definition 1.1 of Ref. \cite{meng2014b}.  Here $r=\frac{\tr x}{n}$.

\end{rmk}

\section{Quadratic Relations and Energy Formulae}
For the Poisson realization in Theorem \ref{main theorem}, let us assume that 
the elements of the conformal algebra such as $S_{uv}$, $X_z$ and $Y_w$ are realized as functions 
\[
\mathcal S_{u, v}, \quad \mathcal X_z, \quad \mathcal Y_w
\] respectively. We shall use $\mathcal L_u$ to denote $\mathcal S_{ue}$ and $\mathcal L_{u, v}$ to denote
$\frac{1}{2}(\mathcal S_{uv}+\mathcal S_{vu})$.

The main purpose of this section is to list two corollaries of Theorem \ref{main theorem}, one concerning the secondary quadratic relations and one concerning a formula connecting the Hamiltonian to the angular momentum and the Laplace-Runge-Lenz vector.  The proof can be taken verbatim from the last section of Ref. \cite{BM2015}. 

\begin{corollary}\label{QRelations}
Let $e_\alpha$ be an orthonormal basis for ${\mathrm H}_n(\mathbb H)$. In the following we hide the summation sign over $\alpha$ or $\beta$. For the Poisson realization in Theorem \ref{main theorem},
we have the following secondary quadratic relations:
\begin{enumerate}[(i)]
\item $\mathcal  X_{e_\alpha} \mathcal  L_{e_\alpha}=n\mathcal  X_e \mathcal  L_e$, 
$\mathcal  Y_{e_\alpha} \mathcal  L_{e_\alpha} = n\mathcal  Y_e \mathcal  L_e$,
 \\
 \item ${4\over n}\mathcal  L_{e_\alpha, u} \mathcal  L_{e_\alpha}=-\mathcal  X_u\mathcal  Y_e+\mathcal  X_e \mathcal  Y_u$,
\\ 
\item 
$\mathcal  X_{e_\alpha}^2 =n\mathcal  X_e^2$, $\mathcal  Y_{e_\alpha}^2 =n\mathcal  Y_e^2$,
\\
\item ${2\over n} \mathcal  L_{e_\alpha, u} \mathcal  X_{e_\alpha}=-\mathcal  X_u \mathcal  L_e+\mathcal  L_u \mathcal  X_e$, ${2\over n}\mathcal  L_{e_\alpha,u} \mathcal  Y_{e_\alpha}=\mathcal  Y_u \mathcal  L_e-\mathcal  L_u\mathcal  Y_e$,
\\
\item $\mathcal  X_{e_\alpha} \mathcal  Y_{e_\alpha}= n (\mathcal  L_e^2 + \mu^2), 
$
\\
\item $\frac{4}{n^3}\mathcal  L_{e_\alpha, e_\beta}^2=\mathcal  X_e \mathcal  Y_e-\mathcal  L_e^2+\frac{n-2}{n}\mu^2$.
\end{enumerate}

\end{corollary}

\begin{corollary} Let $e_\alpha$ be an orthonormal basis for ${\mathrm H}_n(\mathbb H)$, $L^2={1\over 2}\sum_{\alpha, \beta} \mathcal L_{e_\alpha, e_\beta}^2$, and $A^2=-1+\sum_\alpha \mathcal A_{e_\alpha}^2$. Then the Hamiltonian $H$ satisfies the relation 
\begin{eqnarray}\label{HLA}
-2H\left (L^2- \frac{n^2(n-1)}{4}\mu^2\right ) = \left(\frac{n}{2}\right)^2(n-1-A^2).
\end{eqnarray}
\end{corollary}

\end{document}